\documentclass[11pt,a4paper]{article}
\usepackage[margin=2.5cm]{geometry}
\usepackage[english]{babel}
\usepackage[utf8]{inputenc}
\usepackage{amsmath,amsthm}
\usepackage{graphicx}
\usepackage{tikz}
\usetikzlibrary{arrows,shapes}
\usepackage{subfig}
\usepackage{url}

\newtheorem{definition}{Definition}
\newtheorem{theorem}{Theorem}
\newtheorem{lemma}{Lemma}

\def\rot{\rotatebox}

\tikzstyle{vertex} = [fill,shape=circle,node distance=80pt]
\tikzstyle{edge} = [fill,opacity=.5,fill opacity=.5,line cap=round, line join=round, line width=50pt]
\tikzstyle{elabel} =  [fill,shape=circle,node distance=30pt]
\pgfdeclarelayer{background}
\pgfsetlayers{background,main}
\usetikzlibrary{decorations.markings,scopes}
\newcommand{\globallist}[2]{%
 \global\edef#1{#1#2}%
}
\tikzset{bary markings/.style = {
  decoration = {
   markings,
   mark = between positions 0 and 1 step .1 with
    {
     \edef\number{\pgfkeysvalueof{/pgf/decoration/mark info/sequence number}}
     \coordinate (r\number);
     \globallist\refpoints{r\number=1,}
    }
  },
  postaction = {decorate}
 }
}
\def\refpoints{}

\title{Estimating transmission from genetic and epidemiological data: a metric to compare transmission trees}

\author{Michelle Kendall \and Diepreye Ayabina \and Caroline Colijn} 

\date{}

\begin{document}

\maketitle

\begin{abstract}
Reconstructing who infected whom is a central challenge in analysing epidemiological data. Recently, advances in sequencing technology have led to increasing interest in Bayesian approaches to inferring who infected whom using genetic data from pathogens. The logic behind such approaches is that isolates that are nearly genetically identical are more likely to have been recently transmitted than those that are very different. A number of methods have been developed to perform this inference. However, testing their convergence, examining posterior sets of transmission trees and comparing methods' performance are challenged by the fact that the object of inference -- the transmission tree -- is a complicated discrete structure. We introduce a metric on transmission trees to quantify distances between them. The metric can accommodate trees with unsampled individuals, and highlights differences in the source case and in the number of infections per infector. We illustrate its performance on simple simulated scenarios and on posterior transmission trees from a TB outbreak. We find that the metric reveals where the posterior is sensitive to the priors, and where collections of trees are composed of distinct clusters. We use the metric to define median trees summarising these clusters. Quantitative tools to compare transmission trees to each other will be required for assessing MCMC convergence, exploring posterior trees and benchmarking diverse methods as this field continues to mature. 
\end{abstract}


\section{Introduction} 

Understanding who infected whom is a key task of epidemiology. High quality reconstruction of who infected whom in an outbreak of an infectious disease allows public health workers to determine whether there are individuals or locations causing high numbers of transmission, to identify those individuals at risk, and to determine which individual characteristics are associated with infectiousness. Ultimately, this knowledge leads to improved infection control and outbreak management. However, outbreak reconstruction is time-consuming, expensive and uncertain.  It often must rely on individuals' recollections of those with whom they have had contact, as well as individual health records, locations in which infection may have spread, and so on. Particularly in the case of sexually transmitted infections and blood-borne infection, this information is sensitive and case identification is challenging. For chronic infections, transmission may have occurred a considerable time before diagnosis, making reconstructing transmission even more challenging. 

For these reasons and others, there is considerable interest in using genetic data from rapidly-evolving viruses and even bacteria in outbreak reconstructions. Recent advances in sequencing technology have meant that it is feasible to obtain whole-genome RNA or DNA sequences from pathogens even in real time during outbreaks~\cite{Quick2016-ix,Gardy2015-un}, and these data can be used to perform outbreak reconstructions, or to refine reconstructions based on traditional epidemiology. The central idea behind genomic approaches to outbreak reconstruction is that genetic polymorphisms in viruses or bacteria accrue even in the short time frame of the outbreak; by comparing cases' pathogen sequences, it is possible to refine estimates of who infected whom. For example, if cases A and B were in close contact at a time when A was infectious, epidemiological investigations alone would likely conclude that A infected B, but if the pathogen sequences are very different genetically, it would rule this out and another infector would be sought to explain B's infection. 

However, inference of transmission using genetic sequences is challenging. It relies not only on a knowledge of the likely time between an individual becoming infected and infecting others (the generation time), and on the likely time between becoming infected and seeking treatment (leading to being known to the health care system) -- this information is used in almost any reconstruction of transmission. Incorporating genetic data also requires a model of how mutations occur: at the time of transmission, or continuously throughout the life of the pathogen, and at what rate (clocklike evolution or a more general model). It requires, implicitly or explicitly, a model of the dynamics of the pathogen within and between hosts: is more than one lineage present, and how many pathogen particles are transmitted upon infection? Finally, it is rare that health authorities identify every case in an outbreak, and handling unknown cases raises additional challenges. Ideally, genetic information is integrated with epidemiological and clinical information to obtain the best possible estimates of who infected whom. 

Interest in the statistical tools necessary to solve these problems is growing rapidly, and diverse methods have been developed. These differ in their statistical approach: whether they have an explicit spatial structure; whether they allow multiple introductions of the pathogen into the community being analysed; whether they allow multiple distinct infections of individual hosts; how they handle the population dynamics of the pathogen in the host; whether they use a phylogenetic tree to capture relationships amongst the pathogen sequences; and how they handle the issue of unknown cases and cases without genetic data. Table~\ref{tab:compare} lists some of the available tools with respect to these variations. While there are a number of exemplars illustrating the relationship between genomic data and transmission (examples include~\cite{Walker2013-jn, Gray2011-et, Koser2012-gb, Gardy2011-vs}), we focus on Bayesian inference methods aiming to provide tools for use by the community. 

\begin{table}[htb]
\small
\centering
\begin{tabular}{llccccccc}
Name/author & \rot{70}{Ref.} & \rot{70}{Multi. intro.} & \rot{70}{Multi. sequences} & \rot{70}{In-host} & \rot{70}{Unsampled} & \rot{70}{Bottleneck $> 1$} &\rot{70}{Phylo} \\ \hline
Outbreaker &   \cite{Jombart2014-kx}      & Yes & No & Limited & Yes & No & Not required \\
TransPhylo &  \cite{Didelot2016-sw}     & No & Yes & Yes & Yes & No &Required \\
Lau et al.\ &   \cite{Lau2015-sf}      &Yes & No & No & Yes & No &Not required\\
SCOTTI &    \cite{De_Maio2016-qv}       &No & Yes& Yes & Yes & No &Required\\
Numinnen et al.\ &  \cite{Numminen2014-md}      &No &No &No &No & No & Required\\
Kenah et al.\ & \cite{Kenah2016-cl}       &No & Yes & Limited & No  & No &Required\\
Mollentze et al.\ &   \cite{Mollentze2014-rm}    &Yes & No &No  &Yes & No & Not required\\
Morelli et al.\ &   \cite{Morelli2012-kc}     &No & No&No & No& No & Not required \\
Soubeyrand et al.\ & \cite{Soubeyrand2016-us}      &No &No &No &Yes  & No & Not required \\
Hall et al.\ &  \cite{Hall2015-lc}    & No & Yes & Yes & No & No & Estimated\\
phybreak &  \cite{Klinkenberg2016-gg,Klinkenberg_undated-tf}     & No & Yes & Yes & No &  No & Estimated \\
Trepar \ & \cite{Stadler2013} & No & No &No & Yes &  No & Required \\
bitrugs &\cite{bitrugs, Worby2016-ev} & Yes & No &Yes &Yes & No  & Not required \\
Worby et al.\ & \cite{Worby2015-sc} &No & No & No & No & Yes & Not required
\\ \hline
\end{tabular}
  \caption{Some available methods for reconstructing transmission trees using genetic data. `Multi. intro' refers to whether the method accounts for multiple introductions of a pathogen into a community, distinguishing whether all cases are part of one outbreak or several smaller ones. `Multi. sequences' refers to whether the method allows for more than one sequenced isolate per case; often this does not mean multiple distinct infections (re-infection), but only monophyletic clonal instances. `In-host' refers to whether the method admits pathogen diversity within individual hosts; if yes, coalescent or branching events may not correspond to transmission events. `Unsampled' refers to whether there may be inferred cases that were not known to health authorities and not included in the dataset (in contrast to known cases without sequences). `Bottleneck $>1$' refers to whether pathogen diversity can be transmitted (if yes) or whether only one unique sequence is transmitted from case to case (if no). `Phylo' refers to whether a phylogenetic tree is required as an input, estimated alongside transmission, or not used. }
\label{tab:compare}
\end{table}

The data integration needs of this field motivate the use of Bayesian approaches, as they provide a natural framework for integration of covariates such as location, clinical indications of infectiousness and other variables, and avoid the need to use summary statistics of the data. However, by their nature, Bayesian approaches produce a posterior collection of inferred transmission trees alongside posterior distributions of scalar parameters. Understanding the nature of posterior uncertainty in a complex object such as a transmission tree is not straightforward. For example, do posterior estimates group into some trees in which case A is the source and B seeded a set of onward infections, versus trees in which case B is the source and D seeded those onward infections? Do the data support distinct alternative stories of the outbreak, or is the posterior unimodal in the space of transmission trees? Which transmission chains had more unsampled  cases? Typically, the fraction of correctly inferred infectors, or the fraction consistent with an external set of data, is used as a measure of the quality of inferred transmission trees. However, this does not capture `how wrong' the incorrect links are, and does not allow informative comparisons either within a posterior set of trees or of the performance of different methods. In addition, summarising the posterior is typically achieved using the Edmond's consensus tree~\cite{Gibbons1985,Didelot2014,Klinkenberg2016}: a consensus graph is constructed by finding the most common infector for each infectee, and then Edmond's algorithm is used to find the minimum directed spanning tree of this graph. It is therefore possible that such a consensus tree is different in structure from every tree in the posterior, particularly when the trees are quite varied. This limits the ability to effectively summarise the posterior.

Here we develop a metric on the space of transmission trees for a set of infected cases. It allows for unsampled individuals in transmission trees, and is also applicable to other kinds of tree structures. We illustrate the metric using random transmission trees with a simple structure, and find that the metric separates groups of transmission trees in an intuitive and meaningful way. We proceed to analyse posterior collections of transmission trees from a Bayesian inference of transmission from genetic data, and we illustrate how the metric allows us to understand posterior uncertainty and sensitivity to priors. 
Additionally, the metric provides a straightforward way to identify a representative median tree from a collection of trees. Such a median tree has advantages over consensus tree constructions because it is one of the trees from the original collection. 

\section{The metric}

We begin by defining what we mean by a transmission tree. We consider the case in which each individual is infected at most once. For many pathogens it is possible that cases are infected sequentially or even co-infected with different variants, but if this is observed in a set of data, we would denote the multiple infections as distinct, each with a unique infector. Note that we allow for the presence of \emph{unsampled} cases amongst the nodes, that is, individuals who were not known to the health care system during the data-gathering process, but whose presence in the transmission has been inferred.

\begin{definition}
A transmission tree $T=(N,E)$ is a directed graph with nodes $N$ and edges $E$, in which each node corresponds to an infected individual and edges correspond to transmission events. The set of nodes $N=S \cup U$, where $S$ is the set of sampled cases and $U$ is the (possibly empty) set of unsampled cases. A directed edge from node $n_i$ to $n_j$ implies that $n_i$ infected $n_j$. We say that $n_i$ is the `infector' and $n_j$ is the `infectee'. Each node has at most one infector. We require the graph to be connected.
\end{definition}

Since we do not allow for an infectee to have more than one infector, the graph $(N,E)$ has no cycles and is therefore a tree.

\begin{definition}
The \emph{source case} of an outbreak is the unique node in the transmission tree $T=(N,E)$ which has in-degree zero (no infector in $N$).
For any node $n_i \in N$ there is a unique path $p_i$ in $T$ from the source case along directed edges to $n_i$.

The \emph{depth} of node $n_i$ is the number of edges on the path $p_i$; the source case has depth zero.

The \emph{most recent common infector} (MRCI) of two nodes $n_i$ and $n_j$ is the node with the greatest depth which lies on both paths $p_i$ and $p_j$.
Note that if $n_i$ infected $n_j$, or more generally if $n_i$ lies on the path $p_j$, then their MRCI is $n_i$.
For convenience, the MRCI of $n_i$ and $n_i$ is also defined to be $n_i$.

The \emph{descendants} of $n_i$ are the nodes that can be reached following directed paths originating at $n_i$. 
\end{definition}

The requirements that the tree be connected and that each node have at most one infector imply that there is a unique source node, which reflects the fact that we are not modelling multiple distinct introductions of a pathogen into a community. 

\begin{definition}
For a transmission tree $T$ we define the tree vector 
$$
v(T)=(v_{1,1},v_{1,2},\dots,v_{|N|,|N|})
$$ 
where $v_{i,j}$ is the depth of the MRCI of $n_i$ and $n_j$ in $T$.
\end{definition}

To compare different transmission trees for the same infection we propose using the Euclidean distance between tree vectors, as in~\cite{Cardona2013,Kendall2016}. However, although the trees will contain the same set of sampled cases, $S=\{s_1,s_2,\dots,s_|S| \}$, the number of inferred unsampled cases $|U|$, and hence $|N|$, may differ between trees. 
Therefore, to ensure that we are comparing vectors of the same length, we restrict our attention to the vector of sampled cases, 
$$
v|_{S}(T)=(v_{s_1,s_1},v_{s_1,s_2},\dots,v_{s_{|S|},s_{|S|}}).
$$

In practice, we will often wish to compare trees with respect to transmission paths leading to sampled cases, ignoring sets of `trailing' unsampled cases with no sampled descendants. 
Indeed, many tree inference methods only include unsampled cases to make sense of historic infectors of sampled cases.
The tree vector of sampled cases respects this:

\begin{lemma}
\label{lemma:pruning}
Let $T=(N,E)$ be a transmission tree. 
Let $T^*=(N^*,E^*)$ be a copy of $T$, except that any unsampled cases in $T$ without infectees have been pruned (that is, the unsampled case node and its only incident edge removed), and this process repeated until each unsampled case has at least one sampled case somewhere amongst its descendants. 
Then $v|_{S}(T)=v|_{S}(T^*)$.
\end{lemma}

\begin{proof}
The vector $v|_{S}$ records the depths of sampled cases (the $v_{s_i,s_i}$ entries, where $s_i \in S$) and the depths of MRCIs of pairs of sampled cases.
Recall that by `depth' we mean the number of edges (equivalently, the number of nodes minus one) on the unique path from the source case to the node in question. 
Consider an unsampled case $u$ with no sampled case descendants.
Since its removal would not shorten any path between the source case and a sampled node, or even between the source case and the MRCI of any pair of sampled nodes, its existence and position are entirely masked from $v|_S$.
Thus each entry of $v|_{S}(T)$ is unchanged by the pruning of unsampled cases without sampled descendants, and so $v|_{S}(T)=v|_{S}(T^*)$.
\end{proof}

Since we are interested in comparing transmission trees, it is important to establish when we consider two trees to be equivalent. In particular, since the labels of the sampled cases are key to understanding the transmission process, it is important to distinguish between a tree where `case 1 infected case 2' and a tree where `case 2 infected case 1'. However, since any numbering of unsampled cases is arbitrary, the labels of unsampled cases may be safely ignored. 
We will use the following definition:

\begin{definition} 
Consider two transmission trees $T_1 = (N_1=S \cup U_1,E_1)$ and $T_2=(N_2=S \cup U_2,E_2)$, where the set of sampled cases $S$ is the same in each tree.
Let $T_1^*= (N_1^*=S \cup U_1^*,E_1^*)$ and $T_2^*=(N_2^*=S \cup U_2^*,E_2^*)$ be copies of $T_1$ and $T_2$ respectively, but pruned so that every unsampled case has at least one sampled case amongst its descendants, as in Lemma~\ref{lemma:pruning}.

We say that $T_1$ and $T_2$ are \emph{S-isomorphic} if there is an $S$-label-preserving isomorphism from $T_1^*$ to $T_2^*$, that is, a bijective function $\phi: N_1^* \rightarrow N_2^*$ such that $\phi$ is the identity on $S$: 
$$
\phi(s_i)=s_i \text{ for all } s_i \in S
$$ 
and unpruned edges are preserved: 
$$
(n_i,n_j) \in E_1^* \Leftrightarrow (\phi(n_i),\phi(n_j)) \in E_2^* \enspace .
$$ 
\end{definition}

As an example, the two trees in Figure~\ref{fig:ttree_isomorphic} are $S$-isomorphic: arbitrary differences in labelling of unsampled cases $u_1, u_2, \dots$ will not affect our measure of tree difference, nor will the presence of unsampled cases with no sampled descendants. 

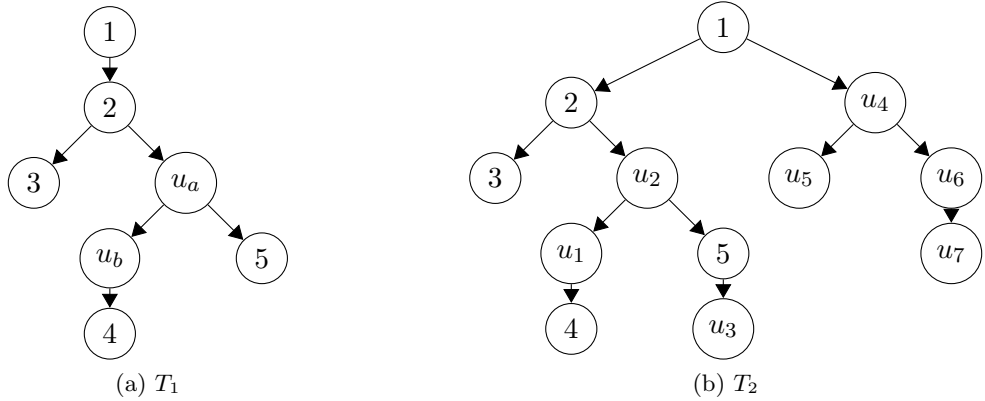
\begin{figure}[htb]
\begin{center}
\subfloat[$T_1$]{
\begin{tikzpicture}[>=triangle 60]

\tikzstyle{every node}=[draw,shape=circle];

\path (2, 4) node (1) {1};
\path (2, 3) node (2) {2};
\path (1, 2) node (3) {3};  
\path (3, 2) node (u2) {$u_a$};  
\path (2, 1) node (u1) {$u_b$}; 
\path (4, 1) node (5) {5};
\path (2, 0) node (4) {4};

\draw [->] (1) -- (2);
\draw [->] (2) -- (3);
\draw [->] (2) -- (u2);
\draw [->] (u2) -- (u1);
\draw [->] (u2) -- (5);
\draw [->] (u1) -- (4);

\end{tikzpicture}
\label{fig:isomorphic1}
}
\hspace{2cm}
\subfloat[$T_2$]{
\begin{tikzpicture}[>=triangle 60]

\tikzstyle{every node}=[draw,shape=circle];

\path (4, 4) node (1) {1};
\path (2, 3) node (2) {2};
\path (1, 2) node (3) {3};  
\path (3, 2) node (u2) {$u_2$};  
\path (2, 1) node (u1) {$u_1$}; 
\path (4, 1) node (5) {5};
\path (2, 0) node (4) {4};
\path (4, 0) node (u3) {$u_3$};
\path (6, 3) node (u4) {$u_4$};
\path (5, 2) node (u5) {$u_5$};
\path (7, 2) node (u6) {$u_6$};
\path (7, 1) node (u7) {$u_7$};

\draw [->] (1) -- (2);
\draw [->] (2) -- (3);
\draw [->] (2) -- (u2);
\draw [->] (u2) -- (u1);
\draw [->] (u2) -- (5);
\draw [->] (5) -- (u3);
\draw [->] (u1) -- (4);
\draw [->] (1) -- (u4);
\draw [->] (u4) -- (u5);
\draw [->] (u4) -- (u6);
\draw [->] (u6) -- (u7);

\end{tikzpicture}
\label{fig:isomorphic2}
}
\end{center}
\caption{$T_1$ and $T_2$ are $S$-isomorphic because $T_2$ will be the same as $T_1$ (up to the relabelling of unsampled cases) after pruning the unsampled cases with no sampled descendants. 
Explicitly, we are using the bijection $\phi: N_1^* \rightarrow N_2^*$ where $\phi(s_i)=s_i$ for $s_i \in S=\{1,2,3,4,5\}$ and $\phi(u_a)=u_2$, $\phi(u_b)=u_1$.} 
\label{fig:ttree_isomorphic}
\end{figure}

\begin{theorem}
Let $S$ be a set of sampled cases and $\mathcal{T}$ a set of transmission trees, each of whose set of nodes contains the set $S$.
Then for any $T_1, T_2 \in \mathcal{T}$, the Euclidean distance between tree vectors:
$$
d(T_1,T_2) = || v|_{S}(T_1) - v|_{S}(T_2) ||
$$ 
is a metric on $\mathcal{T}$ up to $S$-isomorphism.
\end{theorem}

\begin{proof}
The Euclidean distance between vectors is symmetric, non-negative and satisfies the triangle inequality.
To prove that $d$ is a metric we need to show that $d(T_1,T_2)=0$ if and only if $T_1$ and $T_2$ are $S$-isomorphic.

Since the vectors are well-defined and are not conditional on the labelling of unsampled cases, and by Lemma~\ref{lemma:pruning}, we know that when $T_1$ and $T_2$ are $S$-isomorphic then $v|_S(T_1)=v|_S(T_2)$.
It remains to show that $v|_S(T_1)=v|_S(T_2)$ implies that $T_1$ and $T_2$ are $S$-isomorphic.
The proof follows fairly naturally from results in~\cite{Cardona2013} and~\cite{Kendall2016}. 
Here we provide a proof which also supplies some intuition for an algorithm for reconstructing the transmission tree $T$ from the tree vector $v|_S(T)$. 

Let $T_1=(N_1,E_1),T_2=(N_2,E_2) \in \mathcal{T}$ be trees on a set of sampled cases $S$, and suppose that $v|_S(T_1)=v|_S(T_2)$.
First, we consider the simpler case where there are no unsampled nodes in either tree, so $N_1=N_2=S$.
We consider the identity bijection $\phi:N_1 \rightarrow N_2$ with $\phi(n_i)=n_i$ for all $i \in S=N_1=N_2$. 
To show that $T_1$ and $T_2$ are $S$-isomorphic we must show that $\phi$ preserves all edges so that $E_1 = E_2$.

The unique node $n_0$ with $v_{0,0}(T_1)=0$ is the source case in $T_1$, and for each $i \in N_1$, the value $v_{i,i}(T_1)$ gives the depth of node $n_i$ in $T_1$; similarly for $T_2$.
Thus $v|_S(T_1)=v|_S(T_2)$ implies that $T_1$ and $T_2$ have the same source case and that each sampled node is found at the same depth in both trees.
We begin to see how the vector $v|_S(T_1)$ can be used to construct $T_1$: for each depth $\delta$, we can make a list of the nodes at that depth (nodes $n_j$ which satisfy $v_{j,j}(T_1)=\delta$).
In this way we can start to draw our transmission tree as in Figure~\ref{fig:MatrixToTree_Tree1}, where nodes are at the correct depths but directed edges are yet to be placed. 

Now for every $n_i, n_j \in N_1$, there is an edge $(n_i,n_j) \in E_1$ precisely when $n_i$ and $n_j$ are at consecutive depths (without loss of generality, say $n_i$ is at depth $\delta$ and $n_j$ is at depth $\delta + 1$) AND $v_{i,j}(T_1)=\delta$, since this means that $n_i$ is the infector of $n_j$.
Since $v_S(T_1)=v_S(T_2)$, we have that $v_{i,j}(T_1)=v_{i,j}(T_2)$ for all $n_i,n_j \in S$, and so $(n_i,n_j) \in E_1$ if and only if $(\phi(n_i),\phi(n_j))=(n_i,n_j) \in E_2$.
Thus $E_1=E_2$ and $T_1$ and $T_2$ are $S$-isomorphic.

\begin{figure}[htb]
\centering
\subfloat[][]{
\parbox{0.35\textwidth}{
\centering
\vspace{-3cm}
$ M =
\begin{bmatrix}
0 &   &   &   &   &   &   \\
{\color{red} 0}  & 1 &   &   &   &   &   \\
{\color{blue} 0}  & {\color{red} 1}  & 2 &   &   &   &   \\
{\color{blue} 0}  & {\color{red} 1}  & {\color{blue} 1}  & 2 &   &   &   \\
{\color{blue} 0}  & {\color{blue} 1}  & {\color{red} 2}  & {\color{red} 1}  & 3 &   &   \\
{\color{blue} 0}  & {\color{blue} 1}  & {\color{red} 1} & {\color{red} 2}  & {\color{blue} 1}  & 3 &   \\
{\color{blue} 0}  & {\color{blue} 1}  & {\color{red} 1}  & {\color{red} 2}  & {\color{blue} 1}  & {\color{blue} 2}  & 3  \\
\end{bmatrix}
$
}
\label{fig:MatrixToTree_Mat}
} 
\hspace{1cm}
\subfloat[]{
\begin{tikzpicture}[>=triangle 60]

\tikzstyle{every node}=[draw,shape=circle];

\path (2, 4) node (1) {1};
\path (2, 3) node (2) {2};
\path (1.5, 2) node (3) {4};  
\path (2.5, 2) node (4) {3};  
\path (1, 1) node (5) {5};
\path (2, 1) node (6) {6};
\path (3, 1) node (7) {7};

\end{tikzpicture}
\label{fig:MatrixToTree_Tree1}
}
\hspace{1cm}
\subfloat[]{
\begin{tikzpicture}[>=triangle 60]

\tikzstyle{every node}=[draw,shape=circle];

\path (2, 4) node (1) {1};
\path (2, 3) node (2) {2};
\path (1.5, 2) node (3) {3};  
\path (2.5, 2) node (4) {4};  
\path (1, 1) node (5) {5};
\path (2, 1) node (6) {6};
\path (3, 1) node (7) {7};

\draw [->] (1) -- (2);
\draw [->] (2) -- (3);
\draw [->] (2) -- (4);
\draw [->] (3) -- (5);
\draw [->] (4) -- (6);
\draw [->] (4) -- (7);

\end{tikzpicture}
\label{fig:MatrixToTree_Tree2}
}
\caption{For ease of visual notation, we have written the vector $v$ here as a matrix $M$, where $M_{i,j}=v_{i,j}$, and omitted the upper triangle of the matrix because $M$ is symmetric.
The $v_{i,i}$ entries, shown in (a) as the black, diagonal entries of $M$, determine the depths of the nodes in the transmission tree.
We place each node at its appropriate depth (b). 
Transmissions (directed edges) will be placed to point downwards, from one depth to the next. 
It then remains to check the (red) entries of $M$ corresponding to pairs of nodes at consecutive depths, in order to place the edges in the tree.
To draw the transmission tree as a planar graph it may be desirable to rearrange the order of the nodes at each depth; here, we have swapped the order of nodes 3 and 4.
Blue entries of $M$ are not required for this tree reconstruction.}
\label{fig:MatrixToTree}
\end{figure}

Now suppose that $S$ is a strict subset of $N_1,N_2$ (there are some unsampled cases in each tree). 
By Lemma~\ref{lemma:pruning} we know that if there are any unsampled cases in $T_1$ and/or $T_2$ without sampled descendants, then these will not affect the vectors $v|_{S}(T_1), v|_{S}(T_2)$. It remains to show that there is a bijective function $\phi:N_1^* \rightarrow N_2^*$ such that $\phi$ is the identity on $S$ and $(n_i,n_j)\in E_1^*$ if and only if $\phi(n_i,n_j)\in E_2^*$.

If the source case in $T_1$ is an unsampled case then $v_{s_i,s_i}(T_1) > 0 $ for all $s_i \in S$. 
Since $v|_S(T_1)=v|_S(T_2)$ we also have $v_{s_i,s_i}(T_2) > 0 $ for all $s_i \in S$, and so the source case is unsampled in $T_2$ also.
From the first part of the proof we know that any subtree (a connected subset of nodes) of sampled cases $\hat{S} \subseteq S$ which includes the source case must give rise to a unique vector $v|{\hat{s}}$, so that all node depths and edges are determined. 
By extension, any subtree $T|_{\hat{S}}$ of sampled cases $\hat{S}$ whose minimum depth in $T$ is $\delta$ must also be uniquely determined by $v|_{\hat{s}}$, since $v|_{\hat{s}}(T) = v|_{\hat{s}}(T|_{\hat{S}}) + \delta$.
Therefore we know that the identity map $\phi:S \rightarrow S$ preserves all edges within subtrees of sampled cases: for all $s_i,s_j \in S$, $(s_i,s_j) \in E_1^*$ if and only if $(s_i,s_j) \in E_2^*$.

It remains to show that for any path in $T_1$ from the source to a sampled case, an $S$-isomorphic path exists in $T_2$ (a path can be considered as a tree so we are continuing to use the same definition of $S$-isomorphism).
By definition, the path $p_i$ from the source to a sampled node $n_i \in S$ at depth $\delta$ contains a single node at each depth $1,2,\dots,\delta$, and recall that each sampled node has the same depth in $T_1$ and $T_2$. 

Fix a sampled node $n_i$ at depth $\delta \geq 0$ and consider the path to it from the source case in each tree, $p_i(T_1)$ and $p_i(T_2)$ in $T_1$ and $T_2$ respectively.
Consider a depth $x \in \{0,\dots,\delta\}$ and find the node $n_a$ at depth $x$ on $p_i(T_1)$. 
If $n_a$ is a sampled node then $v_{a,a}(T_1)=x=v_{a,a}(T_2)$ and $v_{a,i}(T_1)=x=v_{a,i}(T_2)$, so the same sampled node $n_a$ also appears at depth $x$ on path $p_i(T_2)$.
Now suppose that $n_a$ is an unsampled node in $T_1$, that is, $n_a \in U_1^*$.
Since there is exactly one node at depth $x$ in $T_1$ which has $n_i$ amongst its descendants, and since this node $n_a$ is unsampled, then there can be no sampled node $n_b \in S$ such that both $v_{b,i}(T_1)=x$ and $v_{b,b}(T_1)=x$.
Since the vectors are equal, there can be no node $n_c \in S$ such that both $v_{c,i}(T_2)=x$ and $v_{c,c}(T_2)=x$, and so the node at depth $x$ on path $p_i(T_2)$ is unsampled also. 

Thus each edge $(n_a,n_b)$ on path $p_i(T_1)$ is in $E_1^*$ and has a corresponding edge $(\phi(n_a),\phi(n_b))$ on path $p_i(T_2)$ in $E_2^*$, where $n_a \in S$ if and only if $\phi(n_a)=n_a \in S$, and $n_a \in U_1^*$ if and only if $n_a \in U_2^*$; similarly for $n_b$.
Since this is true for every path from the source to a sampled node, we have shown that $v|_S(T_1) = v|_S(T_2)$ implies that all such paths are $S$-isomorphic in $T_1$ and $T_2$, hence $T_1$ and $T_2$ are $S$-isomorphic.
\end{proof}

Note that this proof illustrates that many of the entries of $v$ are redundant for the reconstruction of the tree, particularly when all cases are sampled, in which case we can ignore any entries $v_{a,b}$ where $n_a$ and $n_b$ are not at consecutive depths. 
In Figure~\ref{fig:MatrixToTree} we only need the diagonal and red entries of $M$ to construct the tree. 
In fact, since we know that the graph is a tree, wherever there are two red entries in a row, only one red entry is strictly necessary in this example for the placement of edges.
Nevertheless, further (blue) entries are needed to understand the relationships across multiple depths when there are unsampled cases, and these `extra' entries also add weight in the comparison of transmission trees and may be useful if this metric was extended to include edge weights.

The existence of a metric on a set of objects enables a variety of further analyses to be performed. 
These include: visualising the pairwise distances between the objects using projections such as multi-dimensional scaling (MDS)~\cite{Cox2000} and cluster analysis, as proposed in the related literature of phylogenetic tree comparison~\cite{Amenta2002,Hillis2005,Holmes2006,Chakerian2012,Berglund2011,Kendall2016}.
Although the metric we have proposed is not convex, barycentric methods can be used to find a representative `central' tree from a set, for example we can find the geometric median tree as proposed in~\cite{Kendall2016}.

Such methods may be used to compare trees: from different input data, taking into account various combinations of metadata; from different inference processes, with variations in their assumptions and settings; and within the same inference process, for example to assess convergence within a Bayesian posterior. Projecting tree-tree distances into two or three dimensions, assessing clustering and finding representative tree(s) can be important for assessing and summarising the performance of inference processes. Additionally, each tree can be compared to a fixed reference tree, for example to assess the success of an inference process in reconstructing the `true' tree from a simulation, or to estimate the effective sample size of discrete tree structures as proposed in~\cite{Lanfear2016}.

\section{Results}

\subsection{Toy examples} 
The metric which we have proposed here detects any differences between trees.
In particular, it highlights differences in the `shape' of the transmission tree (star-like versus single transmission chain, etc.), corresponding to different transmission dynamics. The shape and depth of the tree is largely determined by the number of infectees per infector.
The measure also highlights differences in the attribution of the source case (and in general, differences in historic transmissions are given more emphasis than recent transmission differences). The distance between two trees also depends on the number and relative positions of unsampled cases.

 We tested how well the metric resolves some of these differences using small examples. For each of the following scenarios we generated 1000 transmission trees at random from the set of trees with the given constraints.
We then applied the metric to find the pairwise distances between them, and projected the distances into a two-dimensional plot using MDS.
We use colours and shapes in the plots to highlight key differences between the trees, and to see where these colours do or do not correspond to position in the MDS.
For each scenario we picked the number of infected cases to be small enough so that it was easy to plot and examine the individual trees by eye, and for it to be possible to take a reasonably large sample from the set of all transmission trees of that size, but large enough for there to be a variety of possible tree structures within the given constraints.

\subsubsection{Scenario 1}

For the first scenario we generated random transmission trees under the following constraints: we had exactly eleven sampled cases and no unsampled cases. Each infector was constrained to infect exactly two cases (a binary tree), and the source case was fixed as case 1. 
Under these constraints there are precisely six possible tree `shapes', each admitting a variety of possible transmission trees through the permutation of the remaining ten case labels. 
The key variation in the MDS plot is associated with the height / shape of the tree: in Figure~\ref{fig:scenario1_MDS} we have coloured each point according to the mean value of its tree vector $v$, i.e. the mean of the MRCI depths in the tree. 
Some example trees are also shown: Figure~\ref{fig:scenario1_max} is a tree with the maximum possible depth (mean of $v \approx 1.4$) and Figure~\ref{fig:scenario1_min} is a tree with the minimum possible depth (mean of $v \approx 0.6$), given the above constraints. The metric distinguishes trees composed of one long transmission chain in which each infection gives rise to only one onward-infecting case (and one case who does not infect anyone else), as opposed to more heterogeneous transmission trees in which some individuals cause two onward \emph{infectious} cases. 

\begin{figure}[htb]
\centering
\subfloat[]{
\includegraphics[height=4cm]{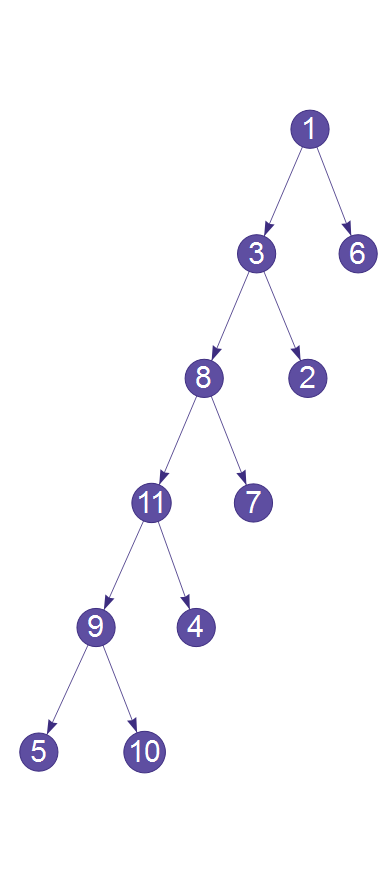}
\label{fig:scenario1_max}
}
\hspace{0.5cm}
\subfloat[]{
\includegraphics[height=4cm]{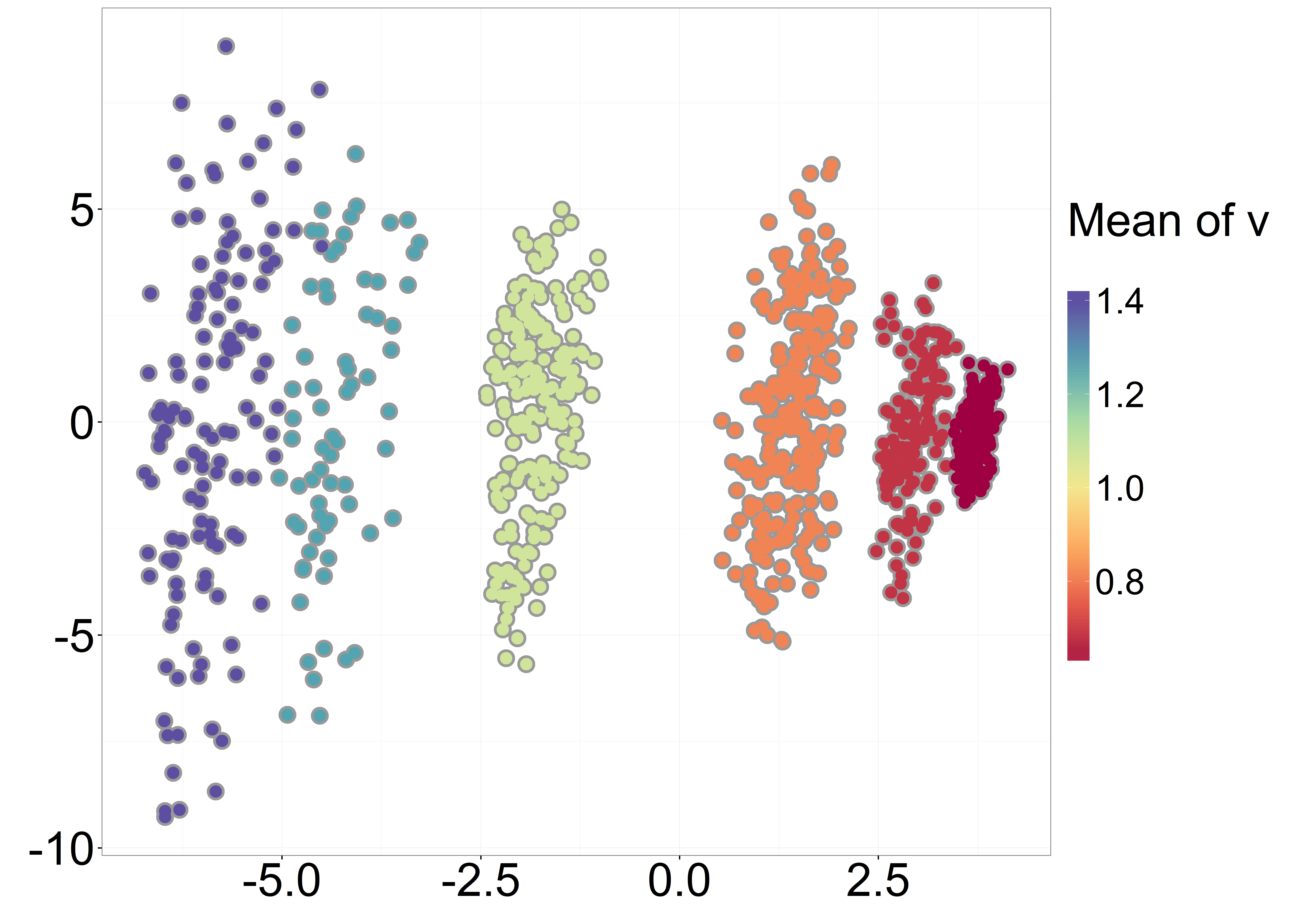}
\label{fig:scenario1_MDS}
}
\subfloat[]{
\includegraphics[height=4cm]{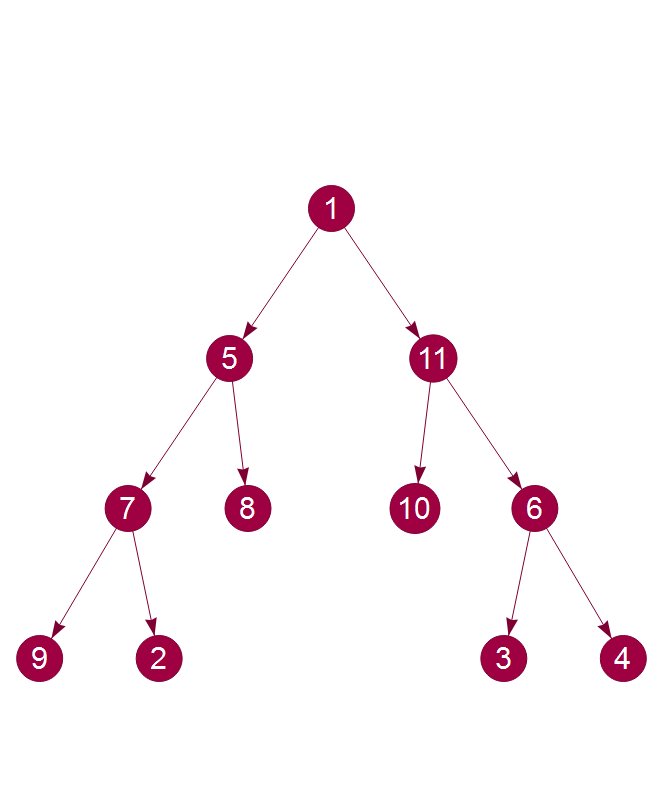}
\label{fig:scenario1_min}
}
\caption{Scenario 1: eleven sampled cases, no unsampled cases, each infector infects exactly two cases, source case fixed as case 1.}
\label{fig:scenario1}
\end{figure}

\subsubsection{Scenario 2}

Our second scenario is similar to the first (eleven sampled cases, no unsampled cases, each infector infects exactly two cases), but now we fix the source case as case 1 in half the trees, and case 2 in the other half.
The resulting MDS plot is shown in Figure~\ref{fig:scenario2}. 
The symmetry in the plot corresponds to the choice of source case (indicated by the shape of each point), and we see that tree shape (shown by colour) is still a discriminating factor. The ability to identify the source case is important for outbreak analysis, and different source cases are likely to correspond to substantially different overall stories of who infected whom.  

\begin{figure}[htb]
\centering
\includegraphics[height=4cm]{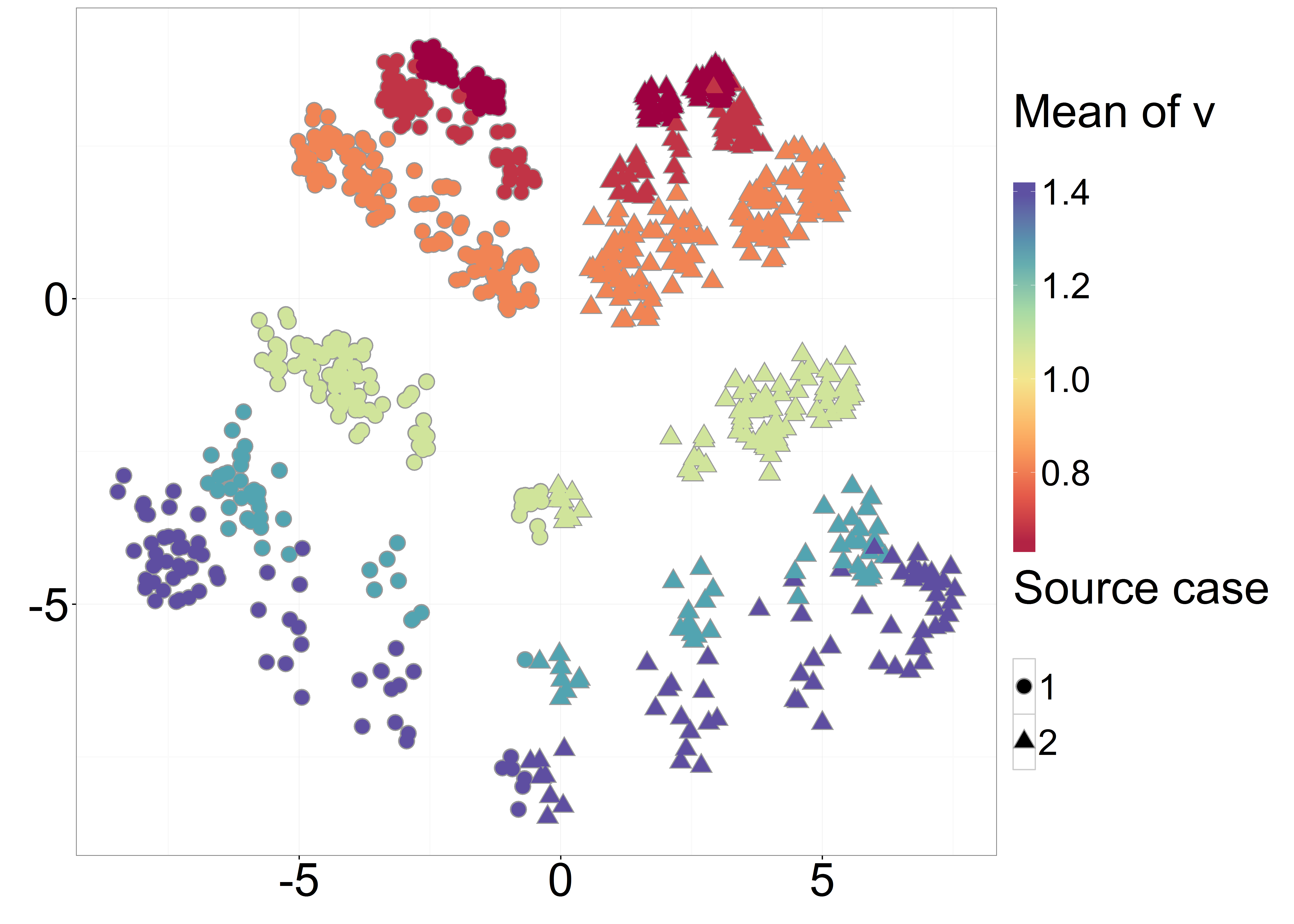}
\caption{Scenario 2: similar to Scenario 1, except half have the source case fixed as case 1, the other half have the source case fixed as case 2. }
\label{fig:scenario2}
\end{figure}

\subsubsection{Scenario 3}

We now reduce the constraint on the number of infectees. For our third scenario, each infector infects $n$ cases, where $n$ is picked uniformly at random from $\{1,2,3\}$, per tree.
Each tree has thirteen sampled cases and no unsampled cases. The source case is picked uniformly at random from $\{1,\dots,6\}$ (for ease of identification by colour in the MDS plots) and the remaining case labels are determined by a random permutation. 
The overwhelming grouping on the first two axes (Figure~\ref{fig:scenario3_axes12}) is by the number of infectees per infector.
In particular, the transmission trees where each infector has one infectee, which are simple chains, are strongly separated from the other trees and are more widely spread in the plot. 
This is because the large number of possible permutations of their labels lead to greater differences in transmission histories than in the shorter, more balanced trees where each infector causes two or three new infections.
There is still some noticeable separation by source case, which becomes much more apparent in a plot of the second and third axes (Figure~\ref{fig:scenario3_axes23}).
This underlines the findings of Scenarios 1 and 2 by showing that the metric distinguishes trees by transmission dynamics and source case attribution, but with rather more emphasis on the former when everything else is fixed.

\begin{figure}[htb]
\centering
\subfloat[MDS plot of axes 1 and 2]{
\includegraphics[height=4cm]{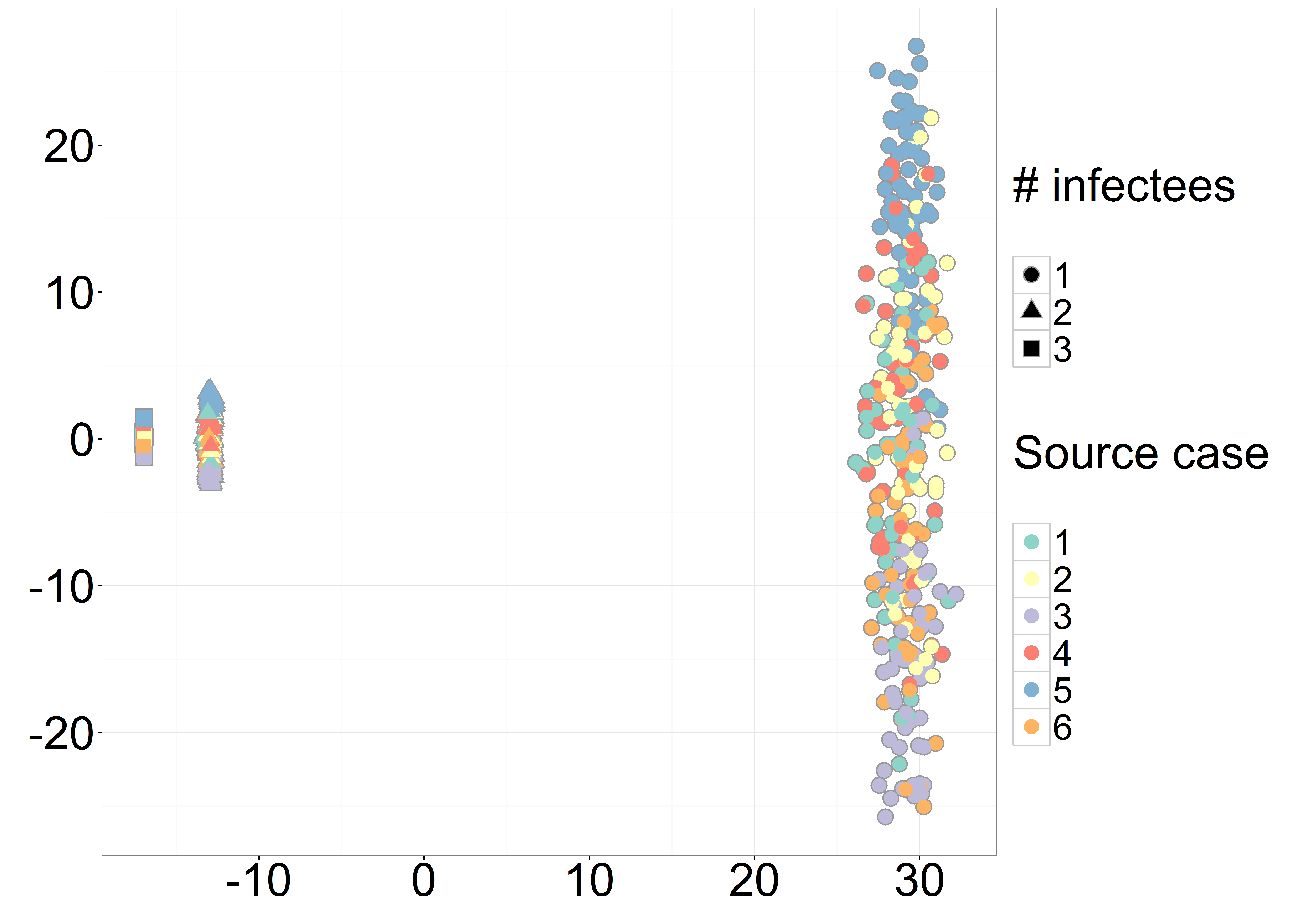}
\label{fig:scenario3_axes12}
}
\hspace{1cm}
\subfloat[MDS plot of axes 2 and 3]{
\includegraphics[height=4cm]{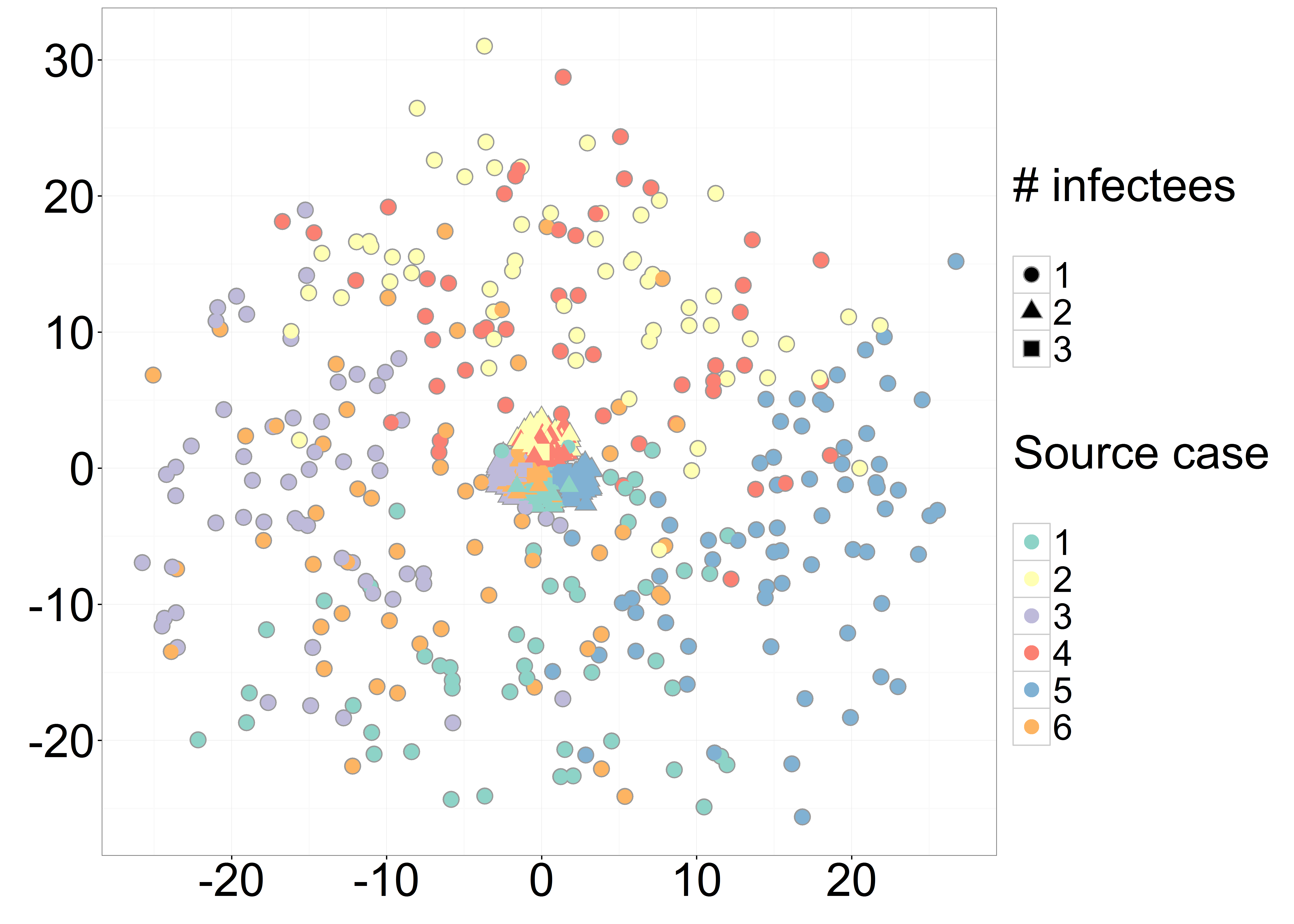}
\label{fig:scenario3_axes23}
}
\caption{Scenario 3: thirteen sampled cases, no unsampled cases, each infector infects $n$ cases, where $n$ is picked uniformly at random from $\{1,2,3\}$, per tree. Source case is picked uniformly at random from $\{1,\dots,6\}$.}
\label{fig:scenario3}
\end{figure}

\subsubsection{Scenario 4}

In our final scenario we analyse the impact of including unsampled cases in our transmission tree. 
We consider trees with eight sampled cases and a further $c$ unsampled cases, where $c$ is picked uniformly at random from $\{0,\dots,8\}$. 
Each infector infects $n$ cases, where $n$ is picked uniformly at random from $\{2,\dots,6\}$, until all cases have been infected (note that this means that not every infector will necessarily infect \emph{exactly} $n$ cases).
Figure~\ref{fig:scenario5} shows how various characteristics of the transmission trees are represented in the MDS plot. 
The first two axes group the trees by features which are correlated with tree shape / transmission dynamics: the mean number of infectees per infector (Figure~\ref{fig:scenario5_meaninfectees}) and the number of unsampled cases in the tree (Figure~\ref{fig:scenario5_numUnsamp}). 
These features are also strongly correlated with the mean of the tree vector $v|_S$ (which captures the depths of \emph{sampled} MRCIs).
As in Scenario 3, there is some grouping by source case (Figure~\ref{fig:scenario5_sourcecase}), particularly by \emph{sampled} source case, especially in the second and third axes (Figure~\ref{fig:scenario5_sourcecase_axes23}), where we have plotted the trees with unsampled source cases with low point opacity. 

\begin{figure}[h!tb]
\begin{center}
\subfloat[Colour: mean number of infectees per infector]{
\includegraphics[height=4cm]{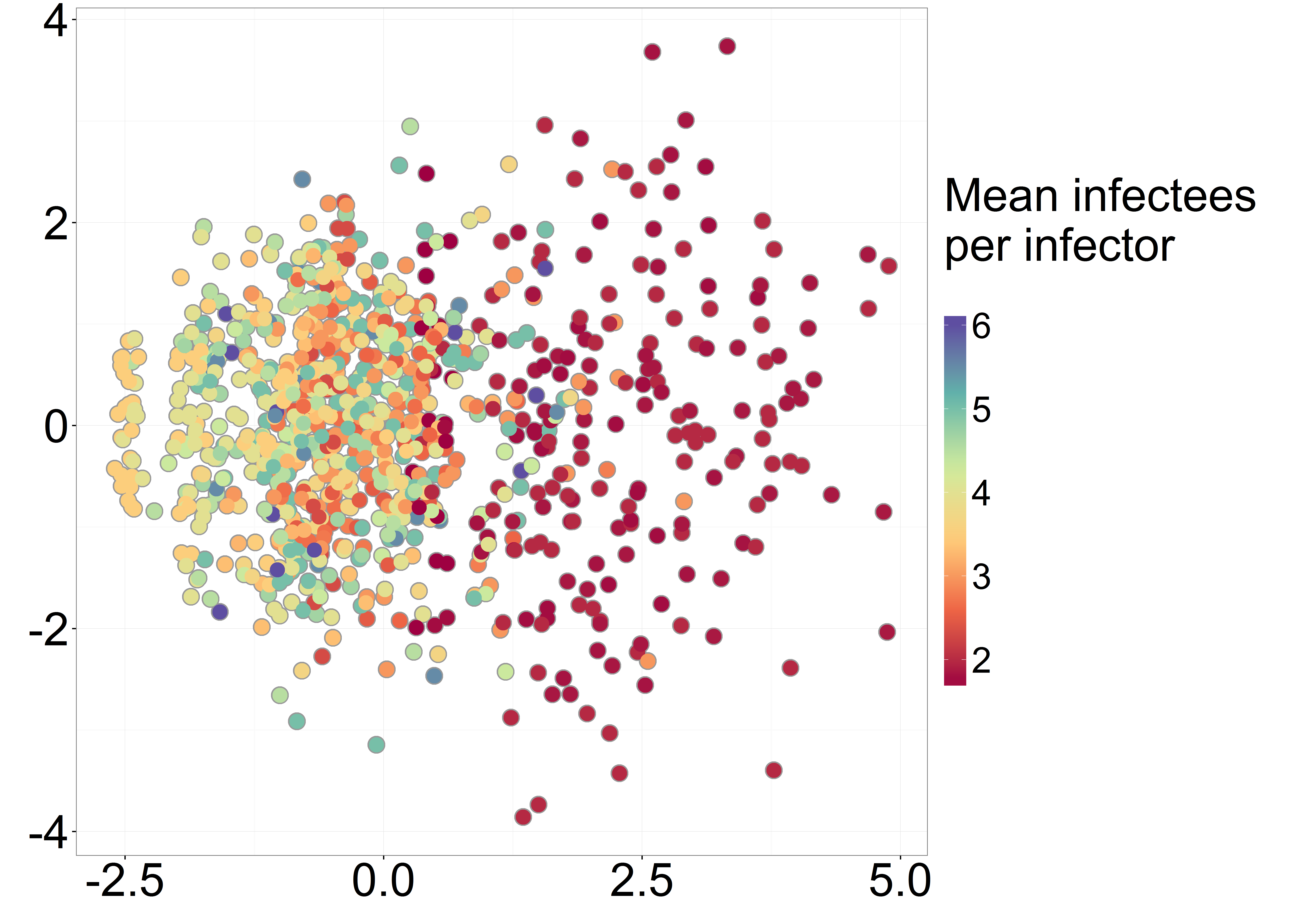}
\label{fig:scenario5_meaninfectees}
}
\hspace{1cm}
\subfloat[Colour: number of unsampled cases]{
\includegraphics[height=4cm]{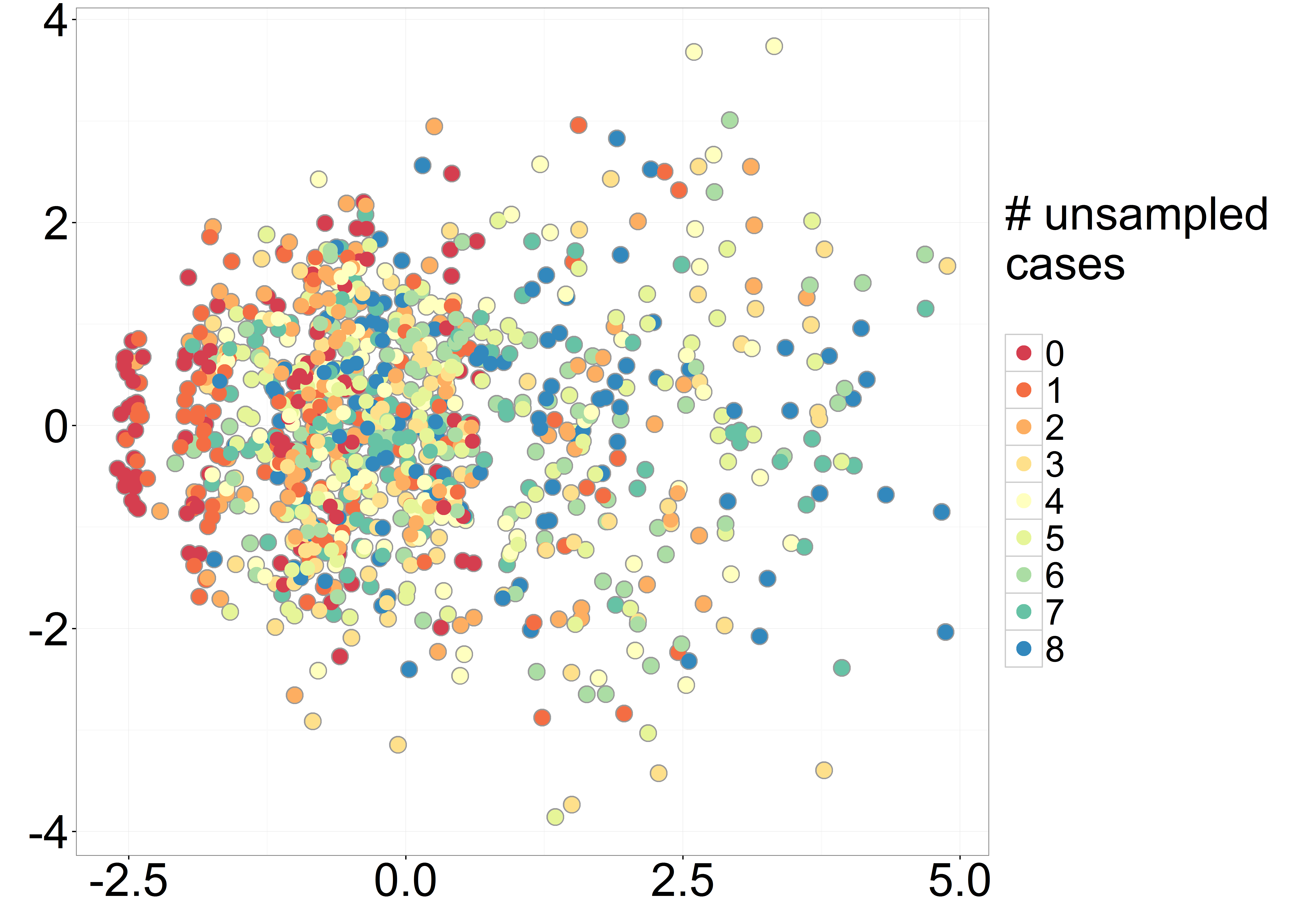}
\label{fig:scenario5_numUnsamp}
} \\
\subfloat[Colour: ID of source case (or $u$ for an unsampled source)]{
\includegraphics[height=4cm]{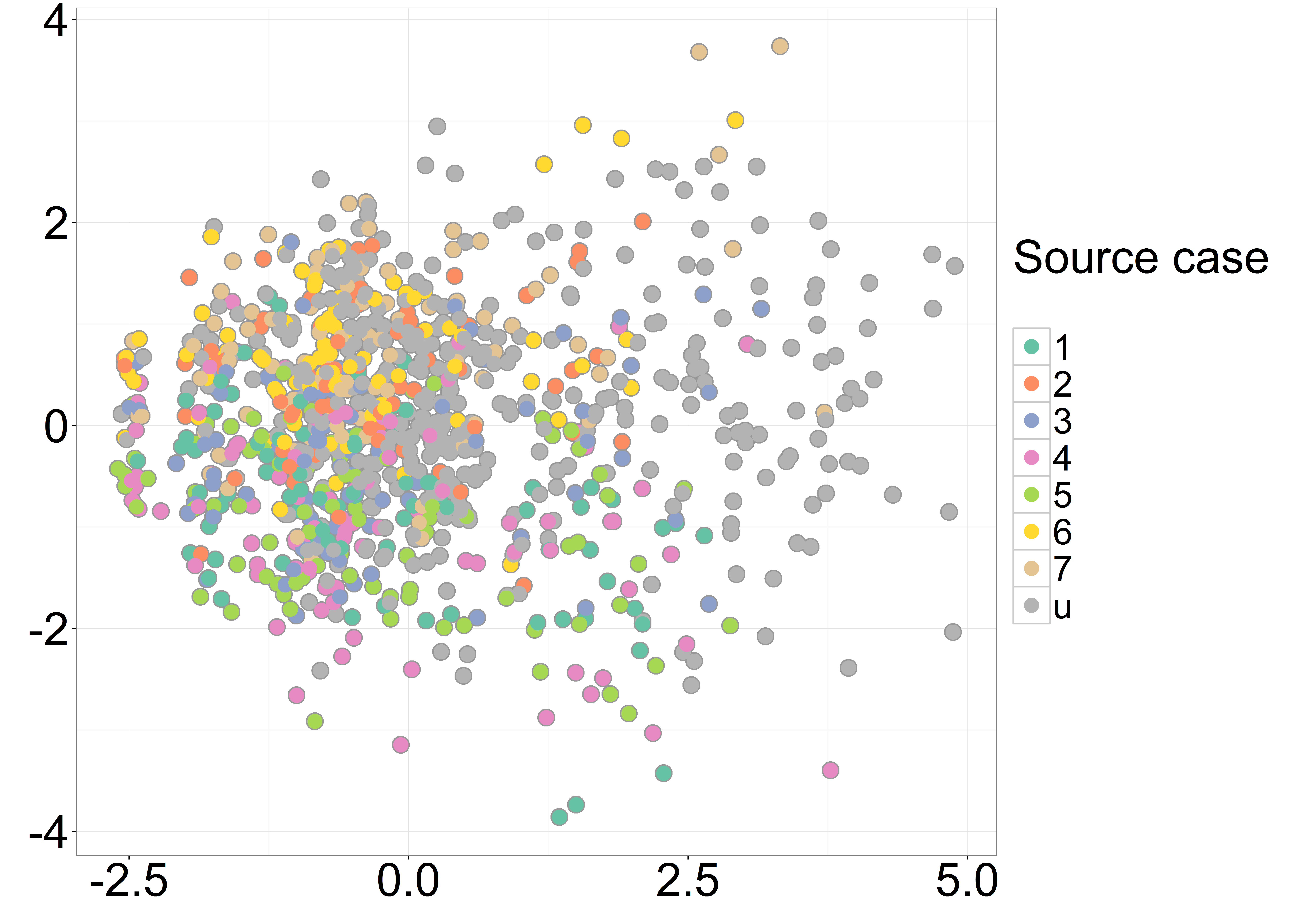}
\label{fig:scenario5_sourcecase}
}
\hspace{1cm}
\subfloat[Axes 2 and 3. Colour: ID of source case (or $u$ for an unsampled source, faded)]{
\includegraphics[height=4cm]{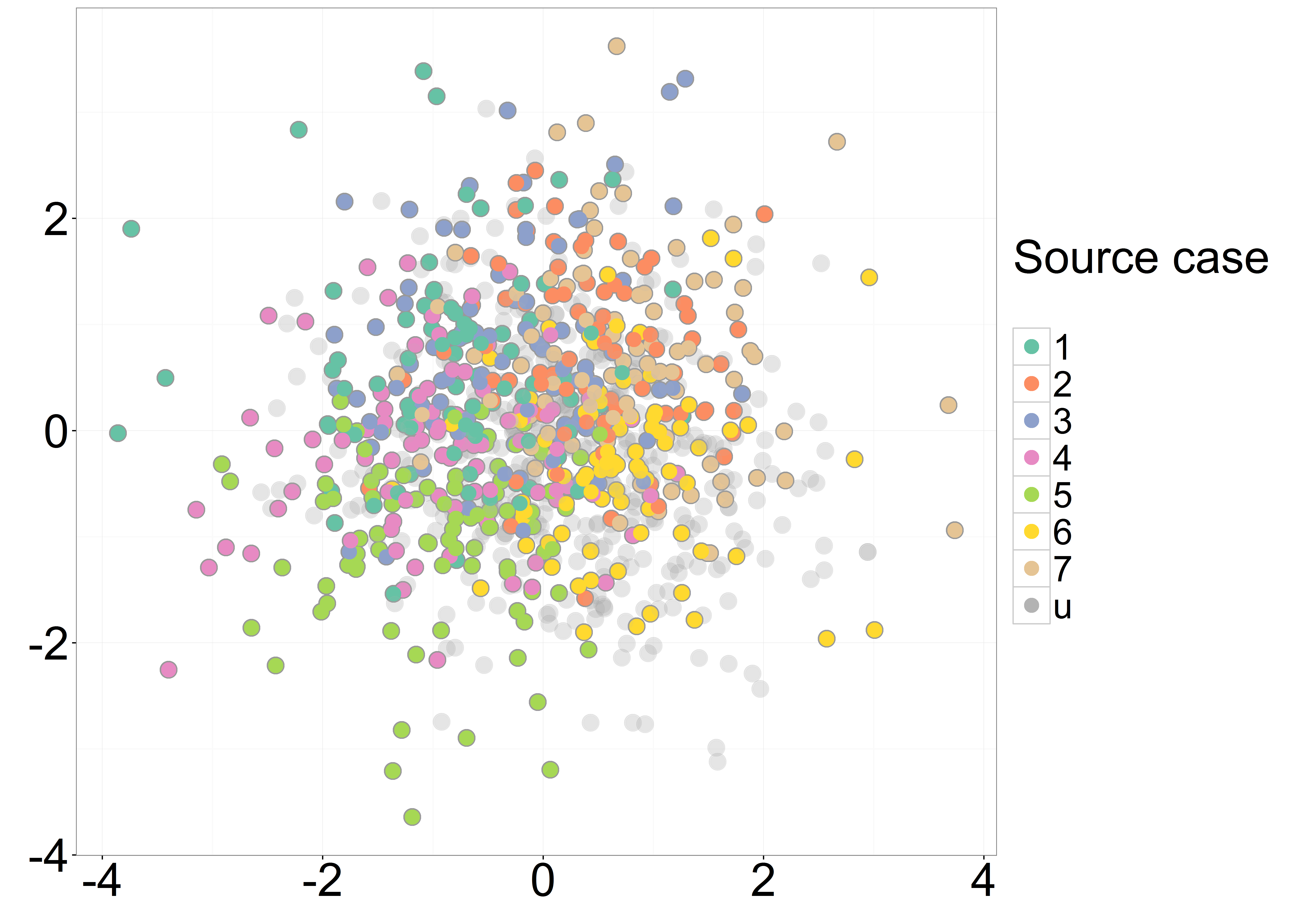}
\label{fig:scenario5_sourcecase_axes23}
}
\caption{MDS plots of tree-tree distances for trees from Scenario 4: eight sampled cases, up to eight further unsampled cases, each infector infects two to six cases. Colour is used to demonstrate how the trees are grouped according to various features. Axes 1 and 2 are plotted except where otherwise stated.}
\label{fig:scenario5}
\end{center}
\end{figure}

\subsection{Tuberculosis outbreak} 
We used the R package TransPhylo~\cite{Didelot2016-sw} to perform MCMC inference to reconstruct an outbreak of tuberculosis (TB) reported by Roetzer et al.~\cite{Roetzer2013-fc}.  The outbreak lasted from 1997 to 2010 during which epidemological data were collected such as information concerning previous exposure to known cases, residence status, sex, and age. TransPhylo is a Bayesian inference method to infer transmission trees using genomic data. TransPhylo's starting point is a timed phylogenetic tree, in which tips correspond to sampled cases and internal nodes correspond to inferred common ancestors; edge lengths are in units of time. The starting tree was inferred using the BEAST~\cite{Drummond2007} software as described in~\cite{Didelot2016-sw}. This tree is held fixed, and TransPhylo proceeds by overlaying transmission events on it, and computing the likelihood of the overall transmission process at each iteration. 

Here we use the metric we have presented to compare inferred transmission trees under different priors, and to explore convergence of the MCMC. The time between an individual becoming infected and infecting others is a major source of uncertainty in TB, as it has a long and variable latent period; this is in contrast to acute infections such as influenza in which the generation time is short and not highly variable (typically under 1-2 weeks). In any public health investigation it is difficult to determine how effectively and rapidly cases are identified. Accordingly, it is important to know how prior assumptions about these distributions affect outbreak reconstructions. The metric allows us to quantify and visualise this. 

We ran $100,000$ MCMC iterations with five different choices for the priors for the sampling and generation times. Some individuals were sampled for reasons other than their symptoms and as such the prior sampling distribution was chosen to be a gamma distribution~\cite{Didelot2016-sw}. Also a gamma distribution was used for the prior generation time distribution in order to reflect the variable disease progression of TB. We sampled $200$ random trees from the last $10000$ iterations of each of the five MCMC runs. We applied the metric to these trees and projected the distances into a two dimensional plot using MDS (Figure~\ref{fig:TB_MDS_plots}). In Figure~\ref{fig:TB_iteration}, we show the distances between the last $1000$ trees from one of the MCMC runs, each tree colored by its iteration number. This reflects how the MCMC moves through the tree space: it samples several times from an area and then hops to another, qualitatively illustrating good mixing.

\begin{figure}[htb]
\centering
\subfloat[Colour: iteration]{
\includegraphics[height=5cm]{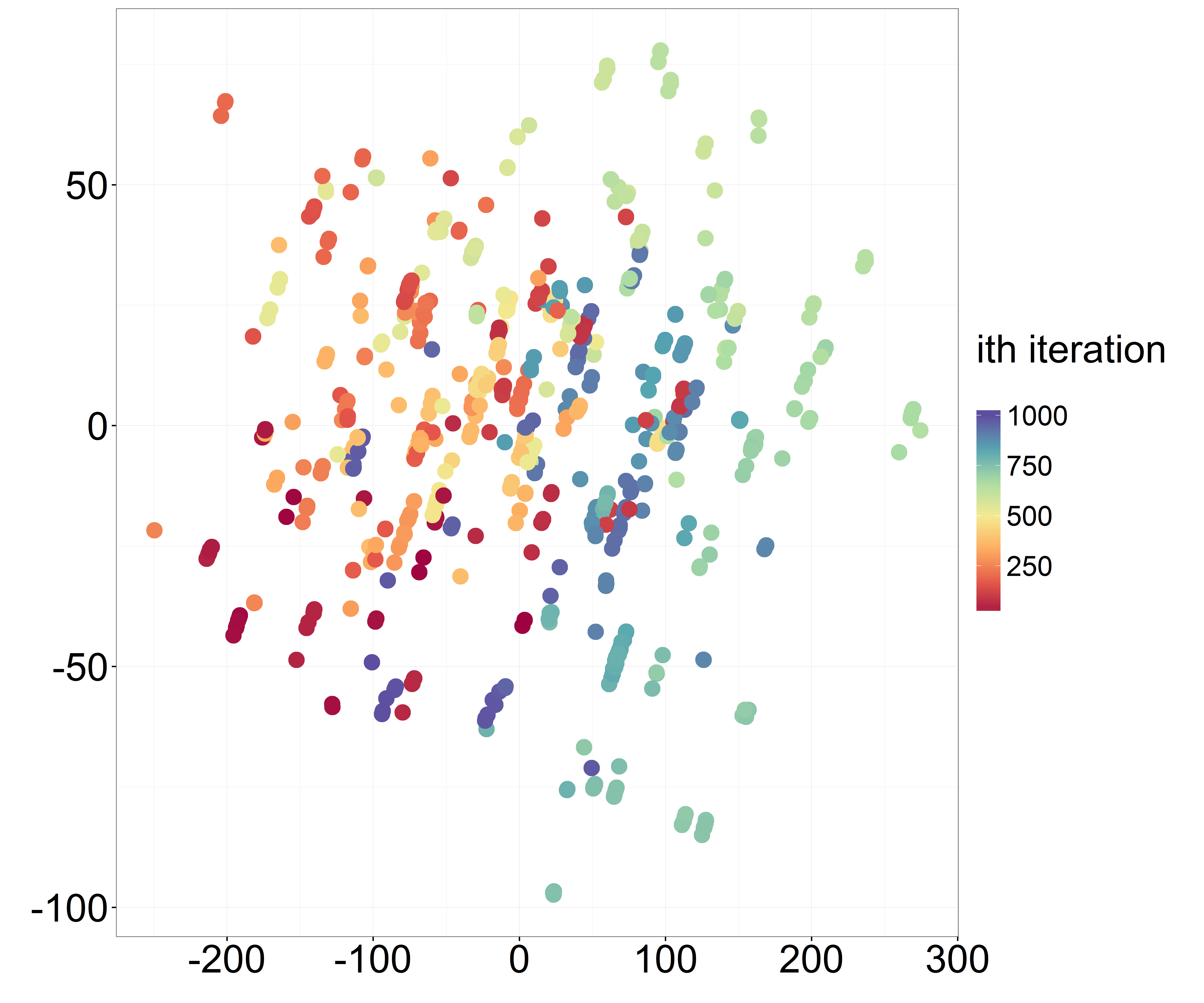}
\label{fig:TB_iteration}
 }
~~
\subfloat[Axes 1 and 2; Colour: mean gen. time; Shape: mean samp. time]{
\includegraphics[height=5cm]{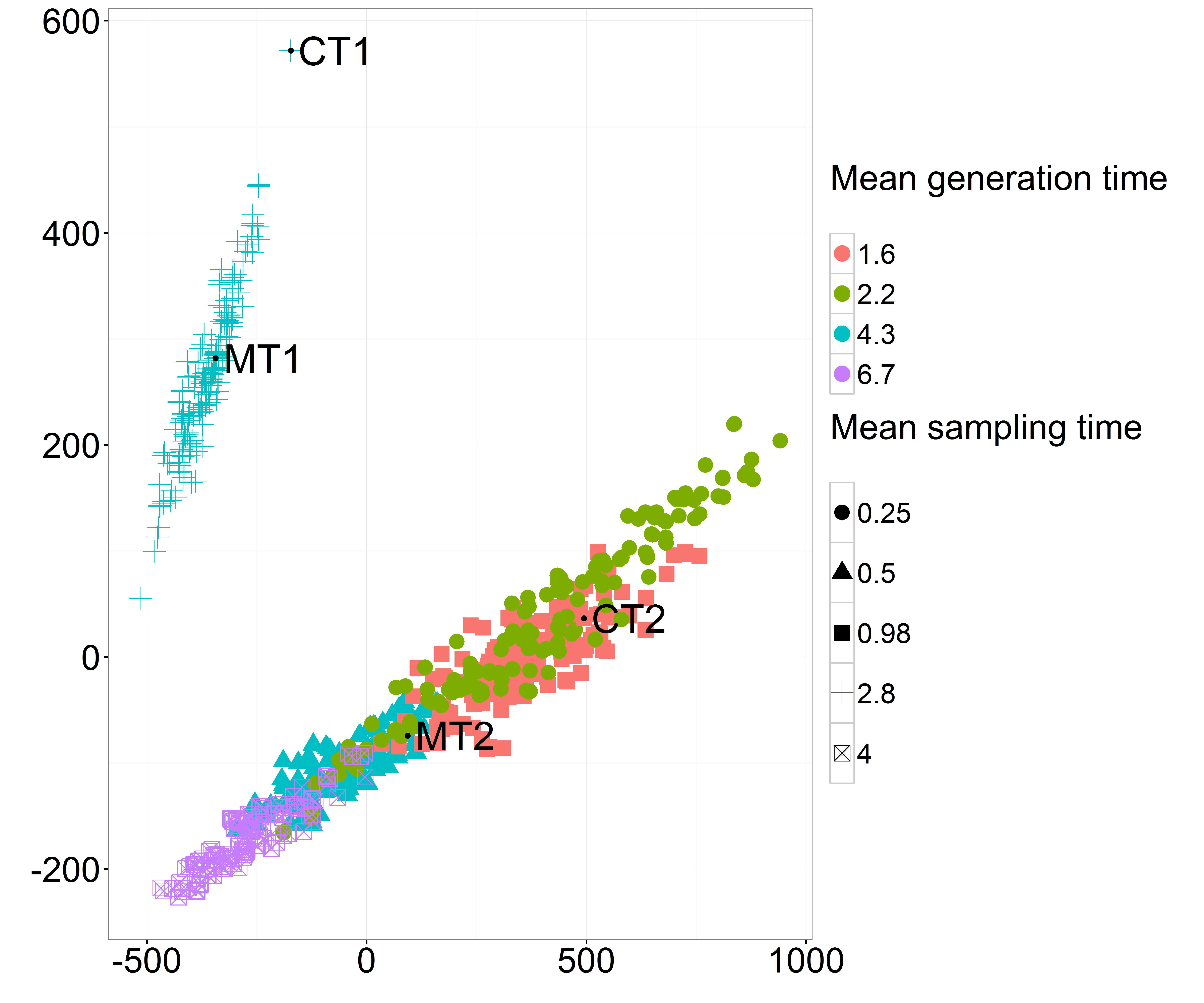}
\label{fig:TB_MDS_12}
} 
\\
\subfloat[Axes 2 and 3; Colour: mean gen. time; Shape: mean samp. time]{
\includegraphics[height=5cm]{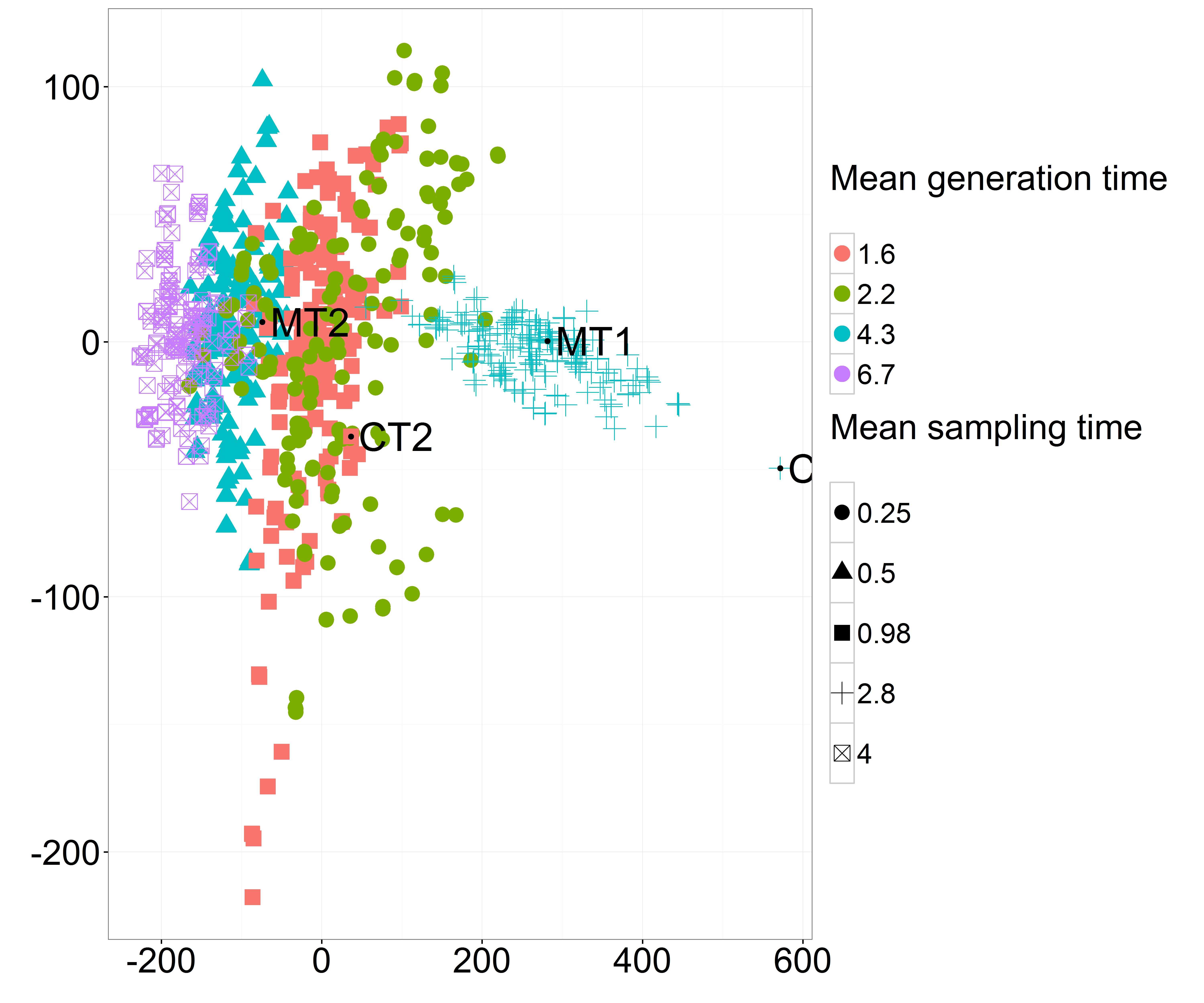}
\label{fig:TB_MDS_23}
}
~~
\caption{MDS plots of tree-tree distances for posterior transmission trees from the Hamburg TB outbreak~\cite{Roetzer2013-fc}. (a) Colour indicates iteration number in the MCMC chain. (b) Colour indicates mean prior generation time, shape indicates mean prior sampling time and the median trees of the two groups are labelled MT1 and MT2.}
\label{fig:TB_MDS_plots}
\end{figure}

Figure~\ref{fig:TB_MDS_plots} illustrates that there are distinct differences between the inferred trees depending on the priors.  Figures \ref{fig:TB_MDS_12} and \ref{fig:TB_MDS_23} show $1000$ trees, 200 from each of the five MCMC runs, on axes $1,2$ and $2,3$ respectively. Colors correspond to mean generation times and shape corresponds to mean sampling times. In Figure \ref{fig:TB_MDS_12}, there are two visually separated clusters of trees. It is not clear why the mean prior generation time of 4.3 years and sampling prior of 2.8 years should produce markedly different trees, as these are not extremal choices of the prior, but in practice it is useful to be able to visualise how unimodal a posterior (or set of trees from multiple posteriors under different priors) is. For the two obvious clusters (blue, and everything else, in the middle panel of Figure~\ref{fig:TB_MDS_plots}), we obtain both a median tree using our metric and a consensus tree using Transphylo's function consTTree which implements Edmond's algorithm. We refer to the smaller blue cluster as cluster 1 and the other as cluster 2. The points $MT1$, $MT2$ correspond to median trees for clusters 1 and 2 while $CT1$ and $CT2$ correspond to (Edmond's) consensus trees of these clusters. $CT1$ is visually separated from the rest of its cluster in the MDS plot, whereas the median trees sit centrally in their clusters. Consistent with this, the mean distances from MT1 and CT1 to trees in cluster 1 are 98 and 306 units respectively. Cluster 2 is larger and more dispersed, and the consensus tree is more central, but the mean distances between MT2 and CT2 and cluster 2's trees are 370 versus 474 units. In our metric the median trees are closer to the clusters they aim to summarise than the trees derived from Edmond's algorithm. The individual transmission trees are illustrated in Figure~\ref{fig:TB_trees}. 

\begin{figure}[htb]
\centering
\includegraphics[width=0.75\textwidth]{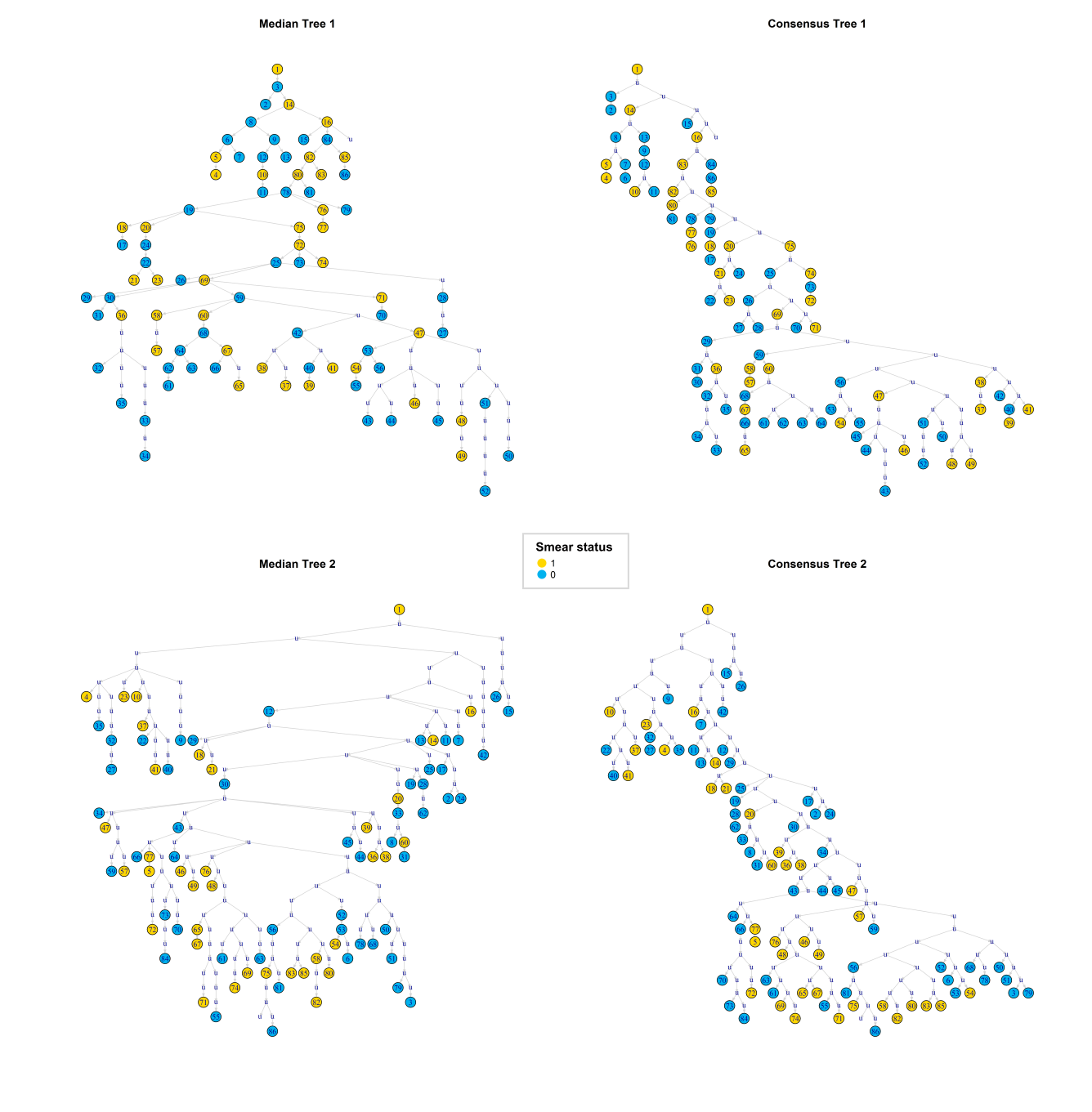}
\caption{Median and consensus trees from each of the two clusters, coloured according to the smear status of each sampled patient.}
\label{fig:TB_trees}
\end{figure}

Trees from the two main clusters have similar depths, and all identify case 1 as the source. Trees from within each cluster have strong similarity in the first few infections after the source case, but there are distinct differences between the clusters, with many individuals placed very differently. For example, note the positions of patients 83 and 85, who appear early in the transmission process in cluster 1 but at the end, with no infectees, in cluster 2. Overall, trees from cluster 2 have more unsampled cases (average 88) than cluster 1 (average 33). This is reflected in the median and consensus trees, with 38 and 60 unsampled cases in MT1 and CT1 versus 145 and 111 in MT2 and CT2 respectively. This is likely a result of the prior assumptions: shorter sampling and generation times (more in cluster 2) use higher numbers of unsampled cases to fill in transmission events along long branches of the fixed phylogenetic tree that is provided as input. 

We visualised the median and consensus trees using colour to indicate patients' TB smear status. The smear status refers to the result of a sputum smear microscopy test, which detects TB bacilli in patient sputum samples. Smear-positive individuals are believed to transmit TB more than smear-negative cases due to the higher numbers of bacilli present in the sputum~\cite{Singh2005-il}, but the smear test itself has limited sensitivity (as low as 50\%) \cite{Siddiqi2003-hy}. In our analysis, smear-positive individuals transmit more in trees MT1 and CT1 than in MT2 and CT2, largely due to the fact that MT2 and CT2 have a much higher fraction of transmission by unsampled cases.

\section{Discussion} 
We have introduced a metric, in the sense of a true distance function, on the set of transmission trees with labelled sampled cases along with unsampled cases (up to our notion of isomorphism). In the context of inferring transmission trees, this metric can aid in assessing convergence, posterior concordance and sensitivity to priors, and in comparing inference methods to each other. It emphasises the source case and the extent of shared transmission events in two trees. We applied the metric to random trees from simple simulated scenarios and found that it can separate trees according to their overall shape, the numbers of infectees per infector, and according to which case is the source. It allows for trees with unsampled cases, an advantage because health authorities rarely know about every case in an outbreak of an infectious disease. 

The metric is sensitive to the source case, and as such, it carries the limitation that trees with different source cases but otherwise similar transmission events may appear a higher distance from each other than intuition would suggest. In addition, while unsampled cases are possible, the metric is only a metric up to pruning of unsampled cases with no descendants, and up to relabelling of unsampled cases.  The way we treat unsampled cases could result in distances that do not always reflect intuition. For example, if one tree has long chains of unsampled cases but otherwise similar connectivity (ie A infects B, versus A infects B via a long chain of intermediate unsampled cases, and this occurs for many pairs of individuals), our metric will show a relatively large distance. If this is not desired in a specific application, the effect can be reduced by collapsing chains of unsampled cases before computing distances. 

The metric as it stands also does not take the timing of transmission events into account, equating for example a tree in which A infects B and then infects C two weeks later, with one in which A infects C and then infects B a year later (as both have A infecting both B and C). It would be straightforward, however, to modify the metric in either of two ways: (1) convert the transmission tree to a genealogical, binary, tree -- capturing pathogen lineages that branch at transmission events -- and then use a metric on those binary trees \cite{Robinson1979b,Billera2001,Kendall2016}, or (2) incorporate timing information in the lengths of branches in the framework we have presented here. In (2), we would construct a vector $w_S(T)$ whose entries were the \emph{time elapsed} between the infection of the MRCIs, rather than the \emph{depths} of the MRCIs, and then the time-sensitive metric could be defined as 
$$
d(T_1,T_2) = ||  \left(\epsilon v|_S(T_1) + (1-\epsilon) w|_S(T_1) \right) - \left(\epsilon v|_S(T_2) + (1-\epsilon) w|_S(T_2) \right) ||  \enspace .
$$ 
With $\epsilon > 0$ this would still be a metric on $\mathcal{T}$.

The metric can be used to aid in computing effective sample sizes for posterior collections of transmission trees. Effective sample sizes (ESS) are routinely used in phylogenetic inference, and should be adopted for inference of transmission trees as well.  Recently, Lanfear et al.~\cite{Lanfear2016} outlined approaches to use distances been phylogenetic tree topologies to compare MCMC runs and assess convergence and autocorrelation -- they used traces of distances between trees along the MCMC chains and a single `focal tree', and distances between trees in the chain sampled at different sampling intervals (`jump distances'). Lanfear et al. computed effective sample sizes by applying standard techniques to distances between posterior trees. The same approaches could be used to estimate effective sample sizes for MCMC chains inferring transmission trees, using the metric we have presented here. 

 The R functions required for the tree distances presented here are available in the \texttt{treescape} package~\cite{treescape}, version 1.10.17 onwards.
 
\section{Concluding remarks} 
Inferring transmission events from epidemiological, clinical and now genetic data is a challenging task, and an important one as understanding transmission is essential for designing the best approaches to control infections. Genomic data are noisy, and the underlying processes generating the true variation are stochastic. However, recent advances in sequencing technologies have led to widespread interest in using pathogen sequences to inform us about who infected whom.  There are now many Bayesian methods available for this inference task, each developed with specific goals and features in mind, and each tested on the authors' own data and  simulation scenario (with~\cite{Klinkenberg2016-gg} as one exception that includes tests on other authors' simulations). 

Understanding convergence, the effects of priors, and the structure of the posterior collections of transmission trees is not trivial.  As this field matures, comparing and benchmarking the performance of different methods will require the ability to quantify how close different approaches come to each other and to gold standard trees that experts agree are the best match to comprehensive data sources for an outbreak. We have developed a metric that can aid in these tasks, illustrated its performance and made it available to the community. 

\bibliographystyle{plain}
\bibliography{library,morerefs} 

\end{document}